\documentclass[reqno]{amsart}
\usepackage{amsmath,amssymb,amsthm}
\usepackage[longnamesfirst,numbers]{natbib}
\usepackage[margin]{fixme}
\usepackage{setspace}
\usepackage{enumerate}         

\newtheorem{thm}{Theorem}[section]

\newtheorem{lem}[thm]{Lemma}
\newtheorem{prop}[thm]{Proposition}
\theoremstyle{definition}
\newtheorem{defn}[thm]{Definition}

\theoremstyle{remark}
\newtheorem{rem}[thm]{Remark}

\numberwithin{equation}{section}

\newcommand{\norm}[1]{\left\Vert#1\right\Vert}

\newcommand{\set}[1]{\left\{#1\right\}}
\newcommand{\Ind}[1]{\mathbf{1}_{\left\{#1\right\}}}
\newcommand{\RR}{\mathbb{R}}

\newcommand{\PP}{\mathbb{P}}
\newcommand{\CC}{\mathbb{C}}

\newcommand{\FF}{\mathbb{F}}

\newcommand{\cU}{\mathcal{U}}

\newcommand{\cF}{\mathcal{F}}

\newcommand{\cQ}{\mathcal{Q}}

\newcommand{\Rplus}{\mathbb{R}_{\geqslant 0}}

\newcommand{\pd}[2]{\frac{\partial #1}{\partial #2}}

\newcommand{\scal}[2]{\left\langle{#1},{#2}\right\rangle}

\renewcommand{\Re}{\mathrm{Re}}

\newcommand{\cD}{\mathcal{D}}

\newcommand{\wh}[1]{{\widehat{#1}}}
\newcommand{\E}[1]{\mathbb{E}\left[#1\right]}                           

\title{The G\"{a}rtner-Ellis theorem, homogenization, and affine processes}
\author{Archil Gulisashvili}
\address{Department of Mathematics, Ohio University, USA}
\email{gulisash@ohio.edu}

\author{Josef Teichmann}
\address{Department of Mathematics, ETH Zurich, Switzerland}
\email{jteichma@math.ethz.ch}

\thanks{We gratefully acknowledge the support by the Institute for Mathematical Research (FIM) and the ETH foundation. We would also like to thank Antoine Jacquier for reading the paper and making very helpful comments.}



\date{}

\begin{document}
\maketitle

\begin{abstract}
We obtain a first order extension of the large deviation estimates in the G\"{a}rtner-Ellis theorem. In addition, for a given family of measures, we find a special family of functions having a similar Laplace principle expansion up to order one to that of the original family of measures. The construction of the special family of functions mentioned above is based on heat kernel expansions. Some of the ideas employed in the paper come from the theory of affine stochastic processes. For instance, we provide an explicit expansion with respect to the homogenization parameter of the rescaled cumulant generating function in the case of a generic continuous affine process.
We also compute the coefficients in the homogenization expansion for the Heston model that is one of the most popular stock price models with stochastic volatility.
\end{abstract}
\section{Introduction}\label{S:I}
The large deviations theory has found numerous applications in mathematical finance (see, e.g., \cite{ph}). For instance, using the methods of the large deviations theory, one can estimate various important characteristics of financial models such as tails of asset price distributions, option pricing functions, and the implied volatility 
(see, e.g., \cite{forjac:09,fj,fjm,fk,fjl,fk,gl,JM} and the references therein). A popular source of information on the large deviations theory is the book \cite{demzei:98} by Dembo and Zeitouni. A useful result in the theory is the G\"{a}rtner-Ellis theorem (see \cite{ellis,gartner}, see also \cite{demzei:98}). This theorem allows to infer the upper and lower estimates in the large deviation principle knowing the properties of the limiting cumulant generating function. 

We will next provide a brief overview of the contents of the paper. In Section 2, a new notion of Laplace principle equivalent expansions for families of functions and measures is introduced. This notion is motivated by the homogenization expansion of the rescaled cumulant generating function associated with an affine stochastic process $X$, that is, the function $\Lambda$ defined by
$$
\Lambda(\epsilon,u)=\epsilon \log \mathbb{E} \left[ \exp\left\{ - \frac{u}{\epsilon} X_\epsilon  \right\} \right] = \epsilon \log \int_{\mathbb{R}} \exp \left\{- \frac{u}{\epsilon} z \right\} p_\epsilon(dz).
$$
Actually, the homogenization expansion mentioned above is nothing else but the real analytic expansion of the function $\Lambda$ with respect to the parameter $\epsilon$ (see Section \ref{S:homben}).
In Section 3, we gather definitions and known facts from the theory of general affine processes, while in Section 4, the homogenization procedure is described in all details for continuous affine processes. The main general results obtained in the paper are contained in Section 2 (see Theorems \ref{T:misss} and \ref{T:gGE}). Theorem \ref{T:misss} states that for any family of measures on the real line, satisfying the conditions in the G\"{a}rtner-Ellis theorem, and such that the homogenization expansion exists, we can find a special family of functions that is Laplace principle equivalent to the original family of measures. The structure of the function family in Theorem \ref{T:misss} resembles the first two terms in the heat kernel expansions on Riemannian manifolds (notice that we face a degenerate situation here, so we could not apply heat kernel expansion directly). Theorem \ref{T:gGE} is a generalization of the G\"{a}rtner-Ellis theorem. It is shown in Theorem \ref{T:gGE} that under the same conditions as in Theorem \ref{T:misss}, the first order large deviation estimates are valid. Finally, in Section 5, we compute the coefficients in the homogenization expansion for the correlated Heston model
that is one of the most popular stochastic stock price models with stochastic volatility.

\section{Distributions with equivalent Laplace principle expansions}\label{S:Laplace}

Laplace's principle is an asymptotic expansion technique, which allows one to approximate integrals of the form
\begin{equation}
\int_a^b f(z) \exp\left\{ - \frac{\phi(z)}{\epsilon}\right\}dz
\label{E:LP}
\end{equation}
as $ \epsilon \to 0 $. We will next formulate a rather general version of Laplace's principle that will be used in the sequel. Suppose the following conditions hold:
\begin{itemize}
\item The functions $f$ and $\phi$ in (\ref{E:LP}) are continuous on the interval $(a,b)$, and the integral in (\ref{E:LP}) converges absolutely for all $0<\epsilon<\epsilon_0$. 
\item The function $\phi$ has a unique absolute minimum that occurs at $z=z_0$ with $a< z_0< b$.
\item The function $\phi$ is strictly convex in a neighborhood of $z_0$. 
\item The function $\phi$ is four times continuously differentiable in a neighborhood of $z_0$, and 
\begin{equation}
\phi(z)=\phi(z_0)+\sum_{n=2}^4\frac{\partial^n\phi(z_0)}{n!}(z-z_0)^n+O\left((z-z_0)^5\right)
\label{E:diffe}
\end{equation}
as $z\rightarrow z_0$.
\item The formula in (\ref{E:diffe}) can be differentiated. More exactly, the condition
\begin{equation}
\partial\phi(z)=\sum_{n=2}^4\frac{\partial^n\phi(z_0)}{(n-1)!}(z-z_0)^{n-1}+O\left((z-z_0)^4\right),\quad z\rightarrow z_0,
\label{E:Taylor}
\end{equation}
holds.
\item The function $f$ is twice continuously differentiable in a neighborhood of $z_0$, and
\begin{equation}
f(z)=\sum_{n=0}^2\frac{\partial^nf(z_0)}{n!}(z-z_0)^n+O\left((z-z_0)^3\right)
\label{E:Ty}
\end{equation}
as $z\rightarrow z_0$. 
\end{itemize}
Then, as $\epsilon\rightarrow 0$, 
\begin{align}
&\int_a^b f(z) \exp\left\{ - \frac{\phi(z)}{\epsilon}\right\}dz  \nonumber \\
& =\exp\left\{ - \frac{\phi(z_0)}{\epsilon}\right\}  \sqrt{\frac{2 \pi \epsilon }{\partial^2\phi(z_0)}} 
\bigg [ f(z_0)+ \epsilon \bigg( \frac{\partial^2f(z_0)}{2{\partial^2\phi(z_0)}}+\frac{5{(\partial^3\phi(z_0))}^2f(z_0)}
{24{(\partial^2\phi(z_0))}^3} \nonumber \\ & - \frac{\partial^4\phi(z_0)f(z_0)}{8{(\partial^2\phi(z_0))}^2}- \frac{\partial^3\phi(z_0)\partial f(z_0)}{2{(\partial^2\phi(z_0))}^2} \bigg) + O\left(\epsilon^2\right)\bigg]. 
\label{E:LPf}
\end{align}
Formula (\ref{E:LPf}) can be derived by following the proof of Theorem 8.1 in \cite{O}. 

Let us next assume that weaker differentiability restrictions than those listed above are imposed on the functions $f$ and $\phi$:
\begin{itemize}
\item The function $\phi$ is twice continuously differentiable in a neighborhood of $z_0$, and 
\begin{equation}
\phi(z)=\phi(z_0)+\frac{\partial^2\phi(z_0)}{2}(z-z_0)^2+O\left((z-z_0)^3\right)
\label{E:differ}
\end{equation}
as $z\rightarrow z_0$.
\item The formula in (\ref{E:diffe}) can be differentiated. More exactly, the condition
\begin{equation}
\partial\phi(z)=\partial^2\phi(z_0)(z-z_0)+O\left((z-z_0)^2\right)\quad\mbox{as}\quad z\rightarrow z_0
\label{E:di}
\end{equation}
holds.
\item The function $f$ is such that
\begin{equation}
f(z)=f(z_0)+O\left(z-z_0\right)\quad\mbox{as}\quad z\rightarrow z_0. 
\label{E:dii}
\end{equation}
\end{itemize}
Then, as $\epsilon\rightarrow 0$,
\begin{align}
\int_a^b f(z) \exp\left\{ - \frac{\phi(z)}{\epsilon}\right\}dz 
=\exp\left\{ - \frac{\phi(z_0)}{\epsilon}\right\}  \sqrt{\frac{2 \pi \epsilon }{\partial^2\phi(z_0)}} 
\bigg [ f(z_0) + O\left(\epsilon\right)\bigg]. 
\label{E:LPff}
\end{align}
\begin{rem}\label{R:Tay} \rm
Using the Taylor formula, we see that (\ref{E:diffe}), (\ref{E:Taylor}), and (\ref{E:Ty}) hold provided that
the function $f$ is three times continuously differentiable and the function $\phi$ is five times continuously differentiable near $z_0$.
Similarly, (\ref{E:differ}), (\ref{E:di}), and (\ref{E:dii}) hold if $f$ is continuously differentiable and $\phi$ is three times continuously differentiable near $z_0$.
\end{rem}

Let $\bf p\rm=\left\{p_\epsilon\right\}_{\epsilon> 0}$ be a family of probability measures on $ \mathbb{R}$. The following assumption is modeled on the behavior of the family of moment generating functions of the affine process and on the homogenization ideas (see Section 4 for more details):
\begin{equation}
\int_{\mathbb{R}} \exp\left\{ - \frac{u}{\epsilon} z \right\} p_\epsilon(dz)=
\exp\big( \frac{\Lambda^{(0)}(u)}{\epsilon} \big) \exp \big( \Lambda^{(1)}(u) \big) \bigg( 1 + \epsilon \Lambda^{(2)}(u) + O(\epsilon^2)\bigg)
\label{E:missing}
\end{equation}
as $\epsilon\rightarrow 0$, where $\Lambda^{(i)}$, $0\le k\le 2$, are continuous functions on the domain $I$. The big $O$ estimate in (\ref{E:missing}) 
is uniform on all closed intervals contained in $I$.

It is not hard to see that the functions $\Lambda^{(i)}$, $0\le i\le 2$, in (\ref{E:missing}) can be recovered from the following formulas:
\begin{equation}
\Lambda^{(0)}(u)=\lim_{\epsilon \to 0 } \epsilon \log \int_{\mathbb{R}} \exp \left\{- \frac{u}{\epsilon} z \right\} p_\epsilon(dz),
\label{E:ge1}
\end{equation}
\begin{equation}
\exp\left\{\Lambda^{(1)}(u)\right\}=\lim_{\epsilon \to 0 } \; \exp\left\{ \frac{\Lambda^{(0)}(u)}{\epsilon}\right\} \int_{\mathbb{R}}
 \exp \left\{- \frac{u}{\epsilon} z \right\} p_\epsilon(dz), 
\label{E:ge2}
\end{equation}
and
\begin{align}
&\exp\left\{\Lambda^{(1)}(u)\right\} \Lambda^{(2)}(u)  \nonumber \\
&=\lim_{\epsilon \to 0 } \; \frac{1}{\epsilon} \left[\exp\left\{\frac{\Lambda^{(0)}(u)}{\epsilon}\right\} 
\int_{\mathbb{R}} \exp \left\{ -\frac{u}{\epsilon} z \right\} p_\epsilon(dz) - 
\exp\left\{\Lambda^{(1)}(u)\right\}\right].
\label{E:ge3}
\end{align}

It will be assumed throughout the rest of the paper that the conditions in the G\"{a}rtner-Ellis theorem hold. More precisely, we suppose that the following are true:
\begin{itemize}
\item The function $\Lambda^{(0)}$ defined in (\ref{E:ge1}) exists as an extended real number for all $u\in\mathbb{R}$. We denote by $I$ the maximum open interval such that the number $\Lambda^{(0)}(u)$ is finite for all $u\in I$.
\item The point $u=0$ belongs to the interval $I$.
\item The function $\Lambda^{(0)}$ is continuously differentiable on $I$, the derivative $\partial_u\Lambda^{(0)}$ is a strictly increasing function on $I$, and the range of the function $\partial_u\Lambda^{(0)}$ is $\mathbb{R}$.
\end{itemize}

The previous restrictions concern only the function $\Lambda^{(0)}$. By the G\"{a}rtner-Ellis theorem, they imply the validity of the large deviation principle for the family $\bf p\rm$. More information on the G\"{a}rtner-Ellis theorem can be found in \cite{demzei:98}. The existence of the functions $\Lambda^{(1)}$ and $\Lambda^{(2)}$ (these functions are determined from (\ref{E:ge2}) and (\ref{E:ge3}), respectively), signals that certain refinements of large deviation results may be possible.
\begin{rem}\label{R:remis} \rm In the paper \cite{JR} of Jacquier and Roome, an assumption similar to that in (\ref{E:missing}) is imposed on the rescaled cumulant generating function (see (2.1) in \cite{JR}). Moreover, there are more similarities between the assumptions in the present section and those in Section 2 of \cite{JR}. Note that the main results obtained in \cite{JR} concern the asymptotic behavior of forward start options and forward smiles.
\end{rem}

The function $\Lambda^{(0)}$ is strictly convex on $I$. Let us define an appropriate Legendre-Fenchel transform of $ \Lambda^{(0)} $, more precisely, we put
$$
\left[\Lambda^{(0)}\right]^{*}(z)= -\inf_{u\in I} (uz + \Lambda^{(0)}(u)),\quad z\in\mathbb{R}.
$$
It is clear that there exists a unique minimizer $ z \mapsto u^*(z) $ in the problem described above, satisfying
the condition
\begin{equation}
\partial_u\Lambda^{(0)}(u^{*}(z))=-z.
\label{E:obr}
\end{equation}
It follows that
\begin{equation}
\left[\Lambda^{(0)}\right]^{*}(z)=-zu^{*}(z)-\Lambda^{(0)}(u^{*}(z)).
\label{E:noos1}
\end{equation}
Since $\Lambda^{(0)}(0)=0$, we have $\left[\Lambda^{(0)}\right]^{*}(z)\ge 0$. It is well-known that the function 
$\left[\Lambda^{(0)}\right]^{*}$ is strictly convex on $\mathbb{R}$. The previous statements, (\ref{E:obr}), and (\ref{E:noos1}) imply that $\left[\Lambda^{(0)}\right]^{*}(z)=0$ if $z=-\partial_u\Lambda^{(0)}(0)$, and $\left[\Lambda^{(0)}\right]^{*}(z)> 0$ 
if $z\neq-\partial_u\Lambda^{(0)}(0)$.

Next, set 
\begin{equation}
d(z)=\sqrt{2\left[\Lambda^{(0)}\right]^{*}(z)}.
\label{E:recover}
\end{equation}
It is clear that 
\begin{equation}
\frac{d^2(z)}{2}=\left[\Lambda^{(0)}\right]^{*}(z).
\label{E:noo}
\end{equation}
Therefore,
\begin{equation}
d(z)=\sqrt{-2\left[zu^{*}(z)+\Lambda^{(0)}(u^{*}(z))\right]}.
\label{E:reco}
\end{equation}
By the strict convexity of the function $\Lambda^{(0)}$,
$$
\inf_{z\in\mathbb{R}} \left[uz + \frac{d^2(z)}{2}\right] = -\Lambda^{(0)}(u),\quad u\in I.
$$

Let $\bf p\rm$ be a family of Borel probability measures satisfying condition (\ref{E:missing}). Our next goal is to find a special family of functions $\bf f\rm=\left\{f_\epsilon\right\}_{\epsilon> 0}$ on $ \mathbb{R}$, for which the asymptotic behavior of rescaled moment generating functions resembles the behavior described in formula (\ref{E:missing}). It would be tempting to try to find an appropriate family $\bf f\rm$ among the families of functions satisfying the following condition as $\epsilon \to 0$:
\begin{align}
\int_{\mathbb{R}}\exp\left\{ - \frac{u}{\epsilon} z \right\} f_\epsilon(z) dz&=
\exp\left\{\frac{\Lambda^{(0)}(u)}{\epsilon}\right\} \exp\left\{ \Lambda^{(1)}(u)\right\} \nonumber \\
&\quad\times \bigg( 1 + \epsilon \Lambda^{(2)}(u) + O\left(\epsilon^2\right)\bigg)
\label{E:funk}
\end{align}
uniformly on compact subintervals of $I$, where the functions $\Lambda^{(k)}$, $0\le k\le 2$, are the same as in (\ref{E:missing}). 
However, we can not always guarantee 
the existence of the integral on the left-hand side of formula (\ref{E:funk})
due to the lack of control of the tail-behavior of the function $f_{\epsilon}$. The remedy here is to localize the condition in (\ref{E:funk}).
\begin{defn}\label{D:Lpe}
Let $\bf p\rm$ be a family of Borel probability measures such that (\ref{E:missing}) holds. We say that a family 
$\bf f\rm$ of continuous functions on $\mathbb{R}$ is Laplace principle equivalent up to order $1$ to the family $\bf p\rm$ 
provided that the following conditions hold: \\
\\
(i)\,\,For every $n\ge 1$ there exists a proper open subinterval $J_n\subset I$ of the interval $I$ such that as $\epsilon\rightarrow 0$,
\begin{align*}
\int_{-n}^n\exp\left\{ - \frac{u}{\epsilon} z \right\} f_\epsilon(z) dz&=
\exp\left\{ \frac{\Lambda^{(0)}(u)}{\epsilon}\right\} \exp \left\{ \Lambda^{(1)}(u) \right\} \\
&\quad \bigg( 1 + \epsilon \Lambda^{(2)}(u) + O_{n,u}\left(\epsilon^2\right)\bigg)
\end{align*}
for all $u\in J_n$. \\
\\
(ii)\,\,The sequence of intervals $J_n$, $n\ge 1$, is increasing and $\bigcup_{n=1}^{\infty}J_n=I$.
\end{defn}

The next statement explains how to construct the family $\bf f\rm$. The ansatz, defining the structure of the function
$f_{\epsilon}$ in formula (\ref{E:fami}), is based on the classical theory of heat kernel expansions. 
\begin{thm}\label{T:misss}
Let $\bf p\rm$ be a family of Borel probability measures on $ \mathbb{R} $ satisfying (\ref{E:missing}), and suppose the conditions in the G\"{a}rtner-Ellis theorem hold. Suppose also that the function $\Lambda^{(0)}$ is five times continuously differentiable on $I$, the function $\Lambda^{(1)}$ is three times continuously differentiable on $I$,
and the function $\Lambda^{(2)}$ is continuously differentiable on $I$.
Define a family $\bf f\rm$ of functions as follows:
\begin{equation}
f_\epsilon(z)=\frac{1}{\sqrt{2\pi \epsilon}}\exp \left\{-\frac{d^2(z)}{2\epsilon}\right\} (C_0(z)+\epsilon C_1(z)),\quad\epsilon> 0,
\label{E:fami}
\end{equation}
where $d$ is given by (\ref{E:reco}),
$$
C_0(z)=\sqrt{\partial_u^2 \Lambda^{(0)}(u^{*}(z))}\exp\left\{\Lambda^{(1)}(u^{*}(z))\right\},
$$
and
\begin{align*}
&C_1(z)=C_0(z)\,\Lambda^{(2)}(u^{*}(z))- \frac{\partial^2C_0(z)\,\partial^2_u\Lambda^{(0)}(u^{*}(z))}{2}
-\frac{5C_0(z)\,\left[\partial^3_u\Lambda^{(0)}(u^{*}(z))\right]^2}
{24\left[\partial_u^2\Lambda^{(0)}(u^{*}(z))\right]^3} \\
&+\frac{C_0(z)\,\left(3\left[\partial^3_u\Lambda^{(0)}(u^{*}(z))\right]^2-\partial^2_u\Lambda^{(0)}(u^{*}(z))\,\partial^4_u\Lambda^{(0)}(u^{*}(z))\right)}{8\left[\partial^2_u\Lambda^{(0)}(u^{*}(z))\right]^3}+\frac{\partial C_0(z)\,\partial^3_u\Lambda^{(0)}(u^{*}(z))}{2\partial^2_u\Lambda^{(0)}(u^{*}(z))}.
\end{align*}
Then the family $\bf f\rm$ is Laplace principle equivalent up to order $1$ to the family $\bf p\rm$. 
\end{thm}
\begin{proof}
The differentiability restrictions on the functions $\Lambda^{(i)}$, $0\le i\le 2$, in the formulation of Theorem \ref{T:misss} are imposed because otherwise the functions 
$C_0$ and $C_1$ are not defined. Note that the function $z\mapsto u^{*}(z)$ is three times continuously differentiable on the real line. The previous statement easily follows from (\ref{E:obr}).

The proof of Theorem \ref{T:misss} is based on the following construction, which uses Laplace's principle. 
For every $n\ge 1$, we have
\begin{align}
&\int_{-n}^n \exp \left\{-\frac{u}{\epsilon} z \right\}f_{\epsilon}(z)dz \nonumber \\
&=\frac{1}{\sqrt{2\pi \epsilon}}\int_{-n}^n\exp\left\{-\frac{1}{\epsilon}
\left(uz+\frac{d^2(z)}{2}\right)\right\} (C_0(z)+\epsilon C_1(z))dz.
\label{E:temp}
\end{align}
Set
\begin{equation}
\phi_u(z)=uz+\frac{d^2(z)}{2}.
\label{E:temps}
\end{equation}
Laplace's principle will be applied to the family of integrals appearing on the right-hand side of (\ref{E:temp}) twice. The first time, formula (\ref{E:LPf}) with $f=C_0$ and $\phi=\phi_u$ will be used, while for the second time, formula (\ref{E:LPff}) will be used with $f=C_1$ and $\phi=\phi_u$. 

The critical point $z^{*}(u)$ of the function $\phi_u$ given by (\ref{E:temps}) is the solution to the equation $\partial_z\left[\Lambda^{(0)}\right]^{*}(z)=u$. It is not hard to see that $z=z^{*}(u)$ if and only if $u=u^{*}(z)$.
It follows from (\ref{E:obr}) that
\begin{equation}
z^{*}(u)=-\partial\Lambda^{(0)}(u),\quad u\in I.
\label{E:orb}
\end{equation}

The next formulas can be derived using (\ref{E:noo}), (\ref{E:noos1}), (\ref{E:temps}), and (\ref{E:orb}). We have
\begin{equation}
\partial^2_z\phi_u(z^{*}(u))=\frac{1}{\partial^2\Lambda^{(0)}(u)},
\label{E:u11}
\end{equation}
\begin{equation}
\partial^3_z\phi_u(z^{*}(u))=\frac{\partial^3\Lambda^{(0)}(u)}{[\partial^2\Lambda^{(0)}(u)]^3},
\label{E:u12}
\end{equation}
and
\begin{equation}
\partial^4_z\phi_u(z^{*}(u))=\frac{3[\partial^3\Lambda^{(0)}(u)]^2
-\partial^2\Lambda^{(0)}(u)\,\partial^4\Lambda^{(0)}(u)}{[\partial^2\Lambda^{(0)}(u)]^5}.
\label{E:u13}
\end{equation}

Let us define the intervals $J_n$ appearing in Definition \ref{D:Lpe} as follows:
$$
J_n=\{u\in I:z^{*}(u)\in(-n,n)\},\quad n\ge 1.
$$
It is not hard to see that condition $(ii)$ in Definition \ref{D:Lpe} is satisfied. 
Next, using (\ref{E:LPf}) and (\ref{E:temp}), we obtain
\begin{align}
&\int_{-n}^n \exp \left\{-\frac{u}{\epsilon} z \right\}f_{\epsilon}(z)dz=\exp\left\{ - \frac{\phi_u(z^{*}(u))}{\epsilon}\right\}  \sqrt{\frac{1}{\partial^2_z\phi(z^{*}(u))}} \nonumber \\
&\bigg [ C_0(z^{*}(u))+ \epsilon \bigg(C_1(z^{*}(u))+ \frac{\partial^2_zC_0(z^{*}(u))}{2{\partial^2_z\phi_u(z^{*}(u))}}+\frac{5{(\partial^3_z\phi_u(z^{*}(u)))}^2C_0(z^{*}(u))}
{24{(\partial^2_z\phi_u(z^{*}(u)))}^3} - \nonumber \\ & - \frac{\partial^4_z\phi_u(z^{*}(u))C_0(z^{*}(u))}{8{(\partial^2_z\phi_u(z^{*}(u)))}^2}- \frac{\partial^3_z\phi_u(z^{*}(u))\partial_zC_0(z^{*}(u))}{2{(\partial^2_z\phi_u(z^{*}(u)))}^2} \bigg) + O_{n,u}\left(\epsilon^2\right)\bigg]
\label{E:two}
\end{align}
as $\epsilon\rightarrow 0$. Note that the differentiability conditions in Theorem \ref{T:misss} allow us to use formulas
(\ref{E:LPf}) and (\ref{E:LPff}) with the functions $f$ and $\phi$ chosen above.

We will next compare the formulas in (\ref{E:missing}) and (\ref{E:two}). Note that
$$
\phi_u(z^{*}(u))=u\,z^{*}(u)-z^{*}(u)\,u^{*}(z^{*}(u))-\Lambda^{(0)}(u^{*}(z^{*}(u)))=
-\Lambda^{(0)}(u).
$$
This shows that if we choose the function $d$ as in (\ref{E:recover}), then the first factors in formulas
(\ref{E:missing}) and (\ref{E:two}) coincide. Moreover, the functions $C_0$ and $C_1$ have to be chosen so that
\begin{equation}
C_0(z^*(u))=\sqrt{\partial_u^2 \Lambda^{(0)}(u)}\exp(\Lambda^{(1)}(u)) 
\label{E:seq1}
\end{equation}
and
\begin{align}
C_1(z^{*}(u))&=C_0(z^*(u))\,\Lambda^{(2)}(u)- \frac{\partial^2_zC_0(z^{*}(u))}{2{\partial^2_z\phi_u(z^{*}(u))}}-\frac{5{(\partial^3_z\phi_u
(z^{*}(u)))}^2\,C_0(z^{*}(u))}
{24{(\partial^2_z\phi_u(z^{*}(u)))}^3} \nonumber \\ 
& + \frac{\partial^4_z\phi_u(z^{*}(u))\,C_0(z^{*}(u))}{8{(\partial^2_z\phi_u(z^{*}(u)))}^2}+ \frac{\partial^3_z\phi_u(z^{*}(u))\,\partial_zC_0(z^{*}(u))}{2{(\partial^2_z\phi_u(z^{*}(u)))}^2}.
\label{E:seq2}
\end{align}

The representations of the functions $C_0$ and $C_1$ given in Theorem \ref{T:misss} can be obtained by plugging
$u=u^{*}(z)$ into (\ref{E:seq1}) and (\ref{E:seq2}), and simplifying the resulting formulas. Equalities (\ref{E:orb})-(\ref{E:u13}) are taken into account in the simplifications.

This completes the proof of Theorem \ref{T:misss}.
\end{proof}
\begin{rem}\label{R:nevsegda}\rm
We have already established that $\left[\Lambda^{(0)}\right]^{*}(y)\ge 0$ for all $y\in\mathbb{R}$. Since
(\ref{E:obr}) and (\ref{E:noos1}) hold, we have
$$
\partial\left[\Lambda^{(0)}\right]^{*}(y)=-u^{*}(y)
$$
for all $y\in\mathbb{R}$. Hence the infimum of the function $\left[\Lambda^{(0)}\right]^{*}$ on the real line is attained at the point $y$ such that $u^{*}(y)=0$. This point is given by $y=z^{*}(0)=\partial\Lambda^{(0)}(0)$. Moreover,
$$
\inf_{y\in\mathbb{R}}\left[\Lambda^{(0)}\right]^{*}(y)=-\Lambda^{(0)}(0)=0.
$$
\end{rem}
\begin{rem}\label{R:vsegda}\rm
A heuristic conclusion that can be reached using Theorem \ref{T:misss} is that the family $\bf f\rm$ is a small-time approximation to the 
family $\bf p\rm$ in a certain very weak sense.
Finding such approximations is an important problem. We consider our results as first modest steps in going beyond the celebrated G\"{a}rtner-Ellis theorem. 
\end{rem}

The next assertion provides a first order large deviation estimate in the G\"{a}rtner-Ellis theorem for families of measures satisfying condition (\ref{E:missing}). Higher order estimates can also be found, but we do not include them in the present paper. Let $A$ be a bounded Borel set. Denote by $\overline{A}$ the closure of the set $A$, and let 
$$
a^{+}=\sup_{z\in A}\{z\}\quad\mbox{and}\quad a^{-}=\inf_{z\in A}\{z\}. 
$$
Then we have $z^{+},\,z^{-}\in\overline{A}$. 
\begin{thm}\label{T:gGE}
Let $\bf p\rm$ be a family of probability Borel measures on $ \mathbb{R} $ such that (\ref{E:missing}) holds. Suppose also that the function $\Lambda^{(0)}$ is twice continuously differentiable on $I$ and the conditions in the G\"{a}rtner-Ellis theorem hold (see the conditions listed after formula (\ref{E:ge3})).  
Suppose also that $ A\subset\mathbb{R} $ is a bounded Borel set, and $ x \in A $. Then the following are true:
\\
\\
(i)\,\,If $x\ge\partial\Lambda^{(0)}(0)$, then as $\epsilon\rightarrow 0$,
\begin{align}
p_\epsilon(A)&\le \exp \left\{- \frac{\left[\Lambda^{(0)}\right]^{*}(x)-u^*(x)\,(a^{+}-x)}{\epsilon}\right\}\exp \left\{\Lambda^{(1)}(u^*(x))\right\} \nonumber \\
&\quad\times\bigg( 1 + \epsilon \Lambda^{(2)}(u^*(x)) + \mathcal{O}(\epsilon^2) \bigg).
\label{E:ram1}
\end{align}
\vspace{0.1in}
(ii)\,\,If $x<\partial\Lambda^{(0)}(0)$, then as $\epsilon\rightarrow 0$,
\begin{align}
p_\epsilon(A)&\le \exp \left\{- \frac{\left[\Lambda^{(0)}\right]^{*}(x)-|u^*(x)|\,(x-a^{-})}{\epsilon}\right\}\exp \left\{\Lambda^{(1)}(u^*(x))\right\} \nonumber \\
&\quad\times\bigg( 1 + \epsilon \Lambda^{(2)}(u^*(x)) + \mathcal{O}(\epsilon^2) \bigg).
\label{E:ram3}
\end{align}
The big $O$ estimates in (\ref{E:ram1}) and (\ref{E:ram3}) are uniform with respect to $x\in A$.
\end{thm}
\begin{rem}\label{R:rii1} \rm
The conditions $x\ge\partial\Lambda^{(0)}(0)$ and $x<\partial\Lambda^{(0)}(0)$ are equivalent to $u^{*}(x)\ge 0$ and
$u^{*}(x)< 0$, respectively.
\end{rem}
\begin{thm}\label{T:gGEl}
Let $\bf p\rm$ be a family of probability Borel measures on $ \mathbb{R} $ such that (\ref{E:missing}) holds. Suppose also that the function $\Lambda^{(0)}$ is twice continuously differentiable on $I$ and the conditions in the G\"{a}rtner-Ellis theorem hold (see the conditions listed after formula (\ref{E:ge3})).  
Suppose also that $ A\subset\mathbb{R} $ is a bounded open set, and $ x \in A $. Then the following are true: \\
\\
(i)\,\,Let $x\ge\partial\Lambda^{(0)}(0)$. Then there exists a constant $\gamma_A> 0$ depending on the set $A$ such that as $\epsilon\rightarrow 0$,
\begin{align}
p_{\epsilon}(A)&\ge \exp \left\{- \frac{\left[\Lambda^{(0)}\right]^{*}(x)+u^*(x)\,(x-a^{-})}{\epsilon}\right\}
\exp \left\{\Lambda^{(1)}(u^*(x))\right\} \nonumber \\
&\quad\times\left(1-\exp\left\{-\frac{\gamma_A}{\epsilon}\right\}\right)\bigg( 1 + \epsilon \Lambda^{(2)}(u^*(x)) + \mathcal{O}(\epsilon^2) \bigg).
\label{E:ram2}
\end{align}
\vspace{0.1in}
(ii)\,\,If $x<\partial\Lambda^{(0)}(0)$, then as $\epsilon\rightarrow 0$,
\begin{align}
p_{\epsilon}(A)&\ge \exp \left\{- \frac{\left[\Lambda^{(0)}\right]^{*}(x)+|u^*(x)|\,(a^{+}-x)}{\epsilon}\right\}
\exp \left\{\Lambda^{(1)}(u^*(x))\right\} \nonumber \\
&\quad\times\left(1-\exp\left\{-\frac{\gamma_A}{\epsilon}\right\}\right)\bigg( 1 + \epsilon \Lambda^{(2)}(u^*(x)) + \mathcal{O}(\epsilon^2) \bigg). 
\label{E:ram4}
\end{align}
The constant $\gamma_A$ in (\ref{E:ram4}) is the same as in (\ref{E:ram2}), and the big $O$ estimates in (\ref{E:ram2}) and (\ref{E:ram4}) are uniform with respect to $x\in A$.
\end{thm}
\begin{rem}\label{R:ischez} \rm
Note that performing the transformation
$\limsup_{\epsilon\rightarrow 0}\epsilon\log p_{\epsilon}(A)$ in the upper estimates in Theorem \ref{T:gGE}, we obtain the upper estimate in the large deviation principle for any bounded Borel set $A$. This gives a little more than the upper estimate in the G\"{a}rtner-Ellis theorem. However, we should not forget that formula (\ref{E:ram1}) was derived under a stronger restriction (\ref{E:missing}), than in the G\"{a}rtner-Ellis theorem. 
\end{rem}

\it Proof of Theorem \ref{T:gGE}. \rm
We borrow some ideas from the proofs of Cramer's theorem and the G\"artner-Ellis theorem given in \cite{demzei:98}. 
The proofs of the upper estimates in those theorems use Chebyshev's inequality. In our case, due to a special structure of the problem, we can provide a slightly more direct proof.

Suppose the conditions in Theorem \ref{T:gGE} hold, and let $u\in I$ and $\epsilon> 0$.
Then we have
\begin{align}
&\int_A \exp\left\{ -\frac{uz}{\epsilon}\right\} p_\epsilon(dz)\ge p_\epsilon (A)
\inf_{z\in\overline{A}}\left[\exp \left\{- \frac{uz}{\epsilon}\right\}\right].
 \label{E:doublee1}
 \end{align} 
It follows from (\ref{E:doublee1}) that for every $u\in I$ there exists $\xi(u)\in\overline{A}$ such that 
\begin{align*}
p_\epsilon (A) \le &  \exp \left\{ \frac{u\xi(u)}{\epsilon}\right\}\int_A \exp\left\{
 -\frac{uz}{\epsilon}\right\} p_\epsilon(dz)  \\
= &  \exp\left\{ -\frac{\Lambda^{(0)}(u)}{\epsilon}\right\} \int_A \exp \left\{ -\frac{u}{\epsilon} z \right\} p_\epsilon(dz) \\
& \times \exp\left\{ \frac{\Lambda^{(0)}(u)+xu + u(\xi(u) - x)}{\epsilon}\right\}.
\end{align*}
Indeed, we can take $\xi(u)=a^{+}$ if $u\ge 0$ and $\xi(u)=a^{-}$ if $u< 0$.

Next, by plugging $u=u^*(x)$ into the previous equalities and taking into account condition (\ref{E:missing}), we get
\begin{align}
p_\epsilon(A) \le & \exp\left\{-\frac{\Lambda^{(0)}(u^*(x))}{\epsilon}\right\} \int_A \exp\left\{-\frac{u^*(x)}{\epsilon} z\right\} p_\epsilon(dz) \nonumber \\
& \times \exp\left\{- \frac{\left[\Lambda^{(0)}\right]^{*}(x)- u^*(x)(\xi(u^*(x)) - x)}{\epsilon}\right\} \nonumber \\
\le& \exp \left\{ \Lambda^{(1)}(u^*(x)) \right\} \exp\left\{ - \frac{\left[\Lambda^{(0)}\right]^{*}(x) - u^*(x)(\xi(u^*(x)) - x)}{\epsilon}\right\} \nonumber \\
&\quad\times\bigg( 1 + \epsilon \Lambda^{(2)}(u^*(x)) + \mathcal{O}(\epsilon^2) \bigg)
\label{key-equation}
\end{align}
as $\epsilon\rightarrow 0$. Now, it is not hard to see that (\ref{key-equation}) implies Theorem \ref{T:gGE}.
\\
\\
\it Proof of Theorem \ref{T:gGEl}. \rm
The lower bounds given in Theorem \ref{T:gGEl} are more delicate. Here we start with the estimate
$$
\int_A \exp\left\{ -\frac{uz}{\epsilon}\right\} p_\epsilon(dz)\le p_\epsilon (A)
\sup_{z\in\overline{A}}\left[\exp \left\{- \frac{uz}{\epsilon}\right\}\right]
$$
instead of the estimate in (\ref{E:doublee1}). This implies that 
\begin{align*}
p_\epsilon (A) &\ge   \exp \left\{ \frac{u\eta(u)}{\epsilon}\right\}\int_A \exp\left\{
 -\frac{uz}{\epsilon}\right\} p_\epsilon(dz)  \\
= &  \exp\left\{ -\frac{\Lambda^{(0)}(u)}{\epsilon}\right\} \int_A \exp \left\{ -\frac{uz}{\epsilon}\right\} p_\epsilon(dz) \\
& \times \exp\left\{ \frac{\Lambda^{(0)}(u)+xu + u(\eta(u) - x)}{\epsilon}\right\},
\end{align*}
for all $u\in I$, where $\eta(u)=a^{-}$ if $u\ge 0$ and $\eta(u)=a^{+}$ if $u< 0$. Therefore
\begin{align}
p_\epsilon (A)\ge & \exp\left\{-\frac{\Lambda^{(0)}(u^*(x))}{\epsilon}\right\} \int_A \exp\left\{-\frac{u^*(x)}{\epsilon} z\right\} p_\epsilon(dz) \nonumber \\
& \times \exp\left\{- \frac{\left[\Lambda^{(0)}\right]^{*}(x)- u^*(x)(\eta(u^*(x)) - x)}{\epsilon}\right\}.
\label{E:delicat} 
\end{align}

Our next goal is to use the change of measure method. Consider a new family $\bf\widetilde{p}\rm$ of probability measures defined by
\[
\widetilde{p}_\epsilon(dz)= \frac{\exp\left\{-\frac{u^*(x)z}{\epsilon}\right\}p_\epsilon(dz)}
{\int_{\mathbb{R}} \exp\left\{-\frac{u^*(x)z}{\epsilon}\right\} p_\epsilon(dz)}, \quad\epsilon> 0. 
\]
Note that the family $\bf\widetilde{p}\rm$ depends on $x$. Then inequality (\ref{E:delicat}) and condition 
(\ref{E:missing}) imply that
\begin{align}
p_\epsilon(A) \ge & \exp\left\{ -\frac{\Lambda^{(0)}(u^*(x))}{\epsilon}\right\}  \int_{\mathbb{R}} \exp \left\{ -\frac{u^*(x)}{\epsilon} z \right\} p_\epsilon(dz) \, \widetilde{p}_\epsilon(A) \nonumber \\
& \times \exp\left\{ - \frac{[\Lambda^{(0)}]^{*}(x) - u^*(x)(\eta(u^*(x)) - x)}{\epsilon} \right\} \nonumber \\
= & \exp \left\{\Lambda^{(1)}(u^*(x))\right\}\bigg( 1 + \epsilon \Lambda^{(2)}(u^*(x)) + \mathcal{O}(\epsilon^2) \bigg) \widetilde{p}_\epsilon(A) \nonumber \\
& \times \exp\left\{ - \frac{[\Lambda^{(0)}]^{*}(x) - u^*(x)(\eta(u^*(x)) - x)}{\epsilon}\right\} 
\label{E:finally}
\end{align}
as $\epsilon\rightarrow 0$.

We will next estimate the quantity 
\begin{equation}
\widetilde{p}_\epsilon(A)=1-\widetilde{p}_\epsilon(A^c)
\label{E:konf}
\end{equation} 
from below. This will be done using the upper estimate in the G\"{a}rtner-Ellis theorem. Let us denote by $\widetilde{\Lambda}^{(0)}$
the function defined by (\ref{E:ge1}) for the family $\bf\widetilde{p}\rm$ instead of the family $\bf{p}\rm$. Then it is not hard to see that
\begin{equation}
\widetilde{\Lambda}^{(0)}(v)=\Lambda^{(0)}(v+u^*(x))- \Lambda^{(0)}(u^*(x)),\quad v\in\widetilde{I},
\label{E:label}
\end{equation}
where $\widetilde{I}=I-u^{*}(x)$. The function $\widetilde{\Lambda}^{(0)}$ and the interval $\widetilde{I}$ depend on $x$. It is clear that $0\in\widetilde{I}$. Moreover,
\[
\left[\widetilde{\Lambda}^{(0)}\right]^{*}(y) = -\inf_{v\in\widetilde{I}}
\left\{yv+\widetilde{\Lambda}^{(0)}(v)\right\}\ge 0 
\]
Next, taking into account that $A^{c}$ is a closed set, and using the upper large deviations estimate in the G\"{a}rtner-Ellis theorem (see Theorem 2.3.6 in \cite{demzei:98}), we obtain
$$
\limsup_{\epsilon \to 0}\left[\epsilon \log \widetilde{p}_\epsilon(A^c)\right] \leq - \inf_{y \in A^c} \left[\widetilde{\Lambda}^{(0)}\right]^{*}(y).
$$

Set $\delta_A=\inf_{y \in A^c}\left[\widetilde{\Lambda}^{(0)}\right]^{*}(y)$. Using Remark \ref{R:nevsegda} and (\ref{E:label}), we see that the unique infimum of the function $\left[\widetilde{\Lambda}^{(0)}\right]^{*}$ on the real line is attained at the point
$$
y=\partial\left[\widetilde{\Lambda}^{(0)}\right]^{*}(0)=\Lambda^{(0)}(u^{*}(x))=x,
$$
and is equal to zero. Since $x\notin A^c$, and the set $A^c$ is closed, we have $\delta_A> 0$. Therefore, for every $\tau> 0$, there exists $\epsilon_{\tau}> 0$ such that
\begin{equation}
\widetilde{p}_{\epsilon}(A^c)\le\exp\left\{\frac{-\delta_A+\tau}{\epsilon}\right\},\quad 0<\epsilon<\epsilon_{\tau}.
\label{E:konfe}
\end{equation}

Fix any number $\tau> 0$ with $0<\tau<\delta_A$, and set $\gamma_A=\delta_A-\tau$. Then (\ref{E:konf}) and (\ref{E:konfe}) imply the following estimate:
\begin{equation}
\widetilde{p}_{\epsilon}(A)\ge 1-\exp\left\{\frac{-\gamma_A}{\epsilon}\right\},\quad 0<\epsilon<\epsilon_{\tau}.
\label{E:finn}
\end{equation}
Finally, using (\ref{E:finally}) and (\ref{E:finn}), we establish estimate (\ref{E:ram2}).

The proof of Theorem \ref{T:gGE} is thus completed.

\section{Affine processes}\label{S:affine}
Let $D$ be a non-empty Borel subset of the real Euclidian space $\RR^d$, equipped with the Borel $\sigma$-algebra $\cD$, and assume
that the affine hull of $D$ is the full space $\RR^d$. To $D$ we add a
point $\delta$ that serves as a `cemetery state'. Define
\[\wh{D} = D \cup \set{\delta}, \qquad \wh{\cD} = \sigma(\cD, \set{\delta}),\]
and equip $\wh{D}$ with the Alexandrov topology, in which any open
set with a compact complement in $D$ is declared an open
neighborhood of $\delta$.\footnote{Note that the topology of $\wh{\cD}$ enters our assumptions in a subtle way: We require later that $X$ is c\`adl\`ag on $\wh{\cD}$, which is a property for which the topology matters.} Any continuous function $f$ defined on $D$
is extended
to $\wh{D}$ by setting $f(\delta) = 0$.

Let $(\Omega, \cF, \FF)$ be a filtered measurable space, on which a family
$(\PP^x)_{x \in \wh{D}}$ of probability measures is defined, and
assume that $\cF$ is $\PP^x$-complete for all $x \in \wh{D}$ and
that the filtration $\FF$ is right continuous. Finally, let $X$ be a c\`adl\`ag process
taking values in $\wh{D}$, whose transition kernel
$$
 p_t(x,A) = \PP^x (X_t \in A), \qquad (t \ge 0, x \in \wh{D}, A \in \wh{\cD})
$$
is a normal time-homogeneous Markov kernel, for which $\delta$ is
absorbing. That is, $p_t(x,.)$ satisfies the following conditions:
\begin{enumerate}[(a)]
 \item $x \mapsto p_t(x,A)$ is $\wh{\cD}$-measurable for each $(t,A) \in \Rplus \times \wh{\cD}$. 
 \item $p_0(x,\set{x}) = 1$ for all $x \in \wh{D}$,
 \item $p_t(\delta,\set{\delta}) = 1$ for all $t \geq 0$
 \item $ p_t(x,\wh{D}) = 1 $ for all $(t,x) \in \Rplus \times
 \wh{D}$, and
 \item the Chapman-Kolmogorov equation
 \[p_{t+s}(x,d\xi) = \int p_t(y,d\xi) \,p_s(x,dy) \]
  holds for each $t,s \ge 0$ and $(x,d\xi) \in \wh{D} \times \wh{\cD}$.
\end{enumerate}

We equip $\RR^d$ with the canonical inner product $\scal{}{}$, and associate to $D$ the set $\cU \subseteq \CC^d$ defined by
$$
 \cU = \set{u \in \CC^d: \sup_{x \in D} \Re \scal{u}{x} < \infty}.
$$
Note that the set $\cU$ is the set of complex vectors $u$ such that the exponential function $x \mapsto e^{\scal{u}{x}}$ is bounded on $D$. It is easy to see that $\cU$ is a convex cone and always contains the set of purely imaginary vectors $i\RR^d$.

\begin{defn}[Affine processes]\label{Def:affine_process}
A stochastic process $X$ is called \emph{affine} with state space $D$, if the transition kernel $p_t(x,d\xi)$ of $X$ satisfies the following conditions:
\begin{enumerate}[(i)]
 \item It is stochastically continuous, i.e.~$\lim_{s \to t}p_s(x,.) = p_t(x,.)$ weakly for all $t \ge 0, x \in D$.
 \item The Fourier-Laplace transform of the kernel depends on the initial state in the following way: there exist functions $\Phi: \Rplus \times \cU \to \CC$ and $\psi: \Rplus \times \cU \to \CC^d$, such that
\begin{equation}\label{Eq:affine_property}
\int_D e^{\scal{\xi}{u}}p_t(x,d\xi) = \Phi(t,u) \exp(\scal{x}{\psi(t,u)})
\end{equation}
for all $t \in \Rplus$, $x \in D$, and $u \in \cU$.
\end{enumerate}
\end{defn}
\begin{rem}
Note that the previous definition does not specify $\psi(t,u)$ in a unique way. However, there is a natural unique choice for $\psi$ that will be discussed in Prop.~\ref{Prop:Phipsi_properties} below. Also note that as long as $\Phi(t,u)$ is non-zero, there exists $\phi(t,u)$ such that $\Phi(t,u) = e^{\phi(t,u)}$, and equality \eqref{Eq:affine_property} becomes
\begin{equation}\label{Eq:affine_property_small_phi}
\int_D e^{\scal{\xi}{u}}p_t(x,d\xi) = \exp\left\{\phi(t,u) +
\scal{x}{\psi(t,u)}\right\}.
\end{equation}
This is the essentially the definition that was used in \cite{Duffie2003}. Condition (\ref{Eq:affine_property_small_phi}) means that the Fourier-Laplace transform of the transition function is the exponential of an \emph{affine function} of $x$. This fact is usually interpreted as the reason for the name `affine process', even though affine functions also appear in other aspects of affine processes, e.g. in the coefficients of the infinitesimal generator, or in the differentiated semi-martingale characteristics. We prefer to use equality \eqref{Eq:affine_property} instead of equality \eqref{Eq:affine_property_small_phi}, since the former equality leads to a slightly more general definition that avoids the necessity of the a-priori assumption that the left hand side of \eqref{Eq:affine_property} is non-zero for all $t$ and $u$. 
\end{rem}

Before we start exploring the first simple consequences of
Definition~\ref{Def:affine_process}, additional notation will be introduced. For any $u \in \cU$, set
$\sigma(u) := \inf \set{t \ge 0: \Phi(t,u) = 0}$
and
$\cQ := \set{(t,u) \in \Rplus \times \cU: t < \sigma(u)}$,
and let $\phi$ be a function on $\cQ$ such that
\[\Phi(t,u) = e^{\phi(t,u)}\quad\mbox{for all}\quad (t,u) \in \cQ.\]
The uniqueness of $\phi$ will be discussed below. The functions $\phi$ and $\psi$ have the following properties 
(see \cite{KST2011}):

\begin{prop}\label{Prop:Phipsi_properties}
Let $X$ be an affine process on $D$. Then
\begin{enumerate}[(i)]
\item \label{Item:sigma_pos} The condition $\sigma(u) > 0$ holds for any $u \in \cU$.
\item The functions $\phi$ and $\psi$ are uniquely defined on $\cQ$ under the restriction that they are jointly continuous and satisfy $\phi(0,0) = \psi(0,0) = 0$.
\label{Item:uniqueness}
\item The function $\psi$ maps $\cQ$ into $\cU$. \label{Item:phipsi_range}
\item The functions $\phi$ and $\psi$ satisfy the \label{Item:semiflow}\emph{semi-flow} property. For any $u \in \cU$ and $t,s \ge 0$ with $t + s \le \sigma(u)$, the following conditions hold:
\begin{align*}
\phi(t+s,u) &= \phi(t,u) + \phi(s,\psi(t,u)), \quad \phi(0,u) = 0 \\
\psi(t+s,u) &= \psi(t,\psi(s,u)), \phantom{+ \phi(t,u)}\quad \psi(0,u) = u
\end{align*}
\end{enumerate}
\end{prop}
\begin{rem}

In the sequel, the functions $\phi$ and $\psi$ will always be chosen according to
Proposition~\ref{Prop:Phipsi_properties}.
\end{rem}

We now introduce the important notion of \emph{regularity}.
\begin{defn}\label{Def:regular}
An affine process $X$ is called \emph{regular} if the derivatives
\begin{align*}
F(u) = \frac{\partial \phi(t,u)}{\partial t}\Bigg |_{t=0+}, \qquad
R(u) = \frac{\partial \psi(t,u)}{\partial t}\Bigg |_{t=0+}
\end{align*}
exist for all $u \in \cU$ and are continuous at $u=0$.
\end{defn}

The next statement illustrates why the regularity is a crucial property. This statement was originally established by \cite{Duffie2003} for affine processes on the state-space $\RR^n \times \Rplus^m$.
\begin{prop}\label{Prop:LevyK}
 Let $X$ be a \emph{regular} affine process. Then there exist $\RR^d$-vectors $b,\beta^1,\ldots,\beta^d$; $d \times d$-matrices $a,\alpha^1,\ldots,\alpha^d$; real numbers $c,\gamma^1,\ldots,\gamma^d$, and signed Borel measures $m,\mu^1,\ldots,\mu^d$ on $\RR^d \setminus \set{0}$ such that the functions $F(u)$ and $R(u)$ can be represented as follows:
\begin{subequations}\label{Eq:FR_LK_form}
\begin{align}
F(u) &= \frac{1}{2}\scal{u}{a u} + \scal{b}{u} - c + \int_{\RR^d
\setminus \set{0}}{\left(e^{\scal{\xi}{u}} - 1 - \scal{
h(\xi)}{u}\right)\,m(d\xi)}\;,\\
R_i(u) &= \frac{1}{2}\scal{u}{\alpha^i u} + \scal{\beta^i}{u} -
\gamma^i + \int_{\RR^d \setminus \set{0}}{\left(e^{\scal{\xi}{u}} -
1 - \scal{ h(\xi)}{u}\right)\,\mu^i(d\xi)}\;.
\end{align}
\end{subequations}
In the previous formulas, $h(x) = x \Ind{\norm{x} \le 1}$ is a truncation function. In addition, for all $x \in D$, 
the quantities 
\begin{subequations}\label{Eq:AB_nu}
\begin{align}
A(x) &= a + x_1 \alpha^1 + \dotsm + x_d \alpha^d,\\
B(x) &= b + x_1 \beta^1 + \dotsm + x_d \beta^d,\\
C(x) &= c + x_1 \gamma^1 + \dotsm + x_d \gamma^d,\\
\nu(x,d\xi) &= m(d\xi) + x_1 \mu^1(d\xi) + \dotsm + x_d \mu^d(d\xi)
\end{align}
\end{subequations}
have the following properties: $A(x)$ is positive semidefinite, $C(x) \le 0$, and 
$$
\int_{\RR^d \setminus \set{0}}{\left( \norm{\xi}^2 \wedge 1\right)} \nu(x,d\xi) < \infty.
$$
Moreover, for $u \in \cU$ and $t \in [0,\sigma(u))$, the functions $\phi$ and $\psi$ satisfy
the following ordinary differential equations:
    \begin{subequations}\label{Eq:Riccati}
    \begin{align}
    \pd{}{t} \phi(t,u)&=F(\psi(t,u)), \quad \phi(0,u)=0\label{Eq:Riccati_F}\\
  \pd{}{t} \psi(t,u)&=R(\psi(t,u)), \quad \psi(0,u)=u\label{Eq:Riccati_R}.
    \end{align}
    \end{subequations}
\end{prop}
\begin{rem}
 The equations \eqref{Eq:Riccati} are called \emph{generalized Riccati equations}, since they are classical Riccati equations when $m(d\xi)$ = $\mu^i(d\xi) = 0$. Moreover, equations \eqref{Eq:FR_LK_form} and \eqref{Eq:AB_nu} imply that $u \mapsto F(u) + \scal{R(u)}{x}$ is a function of L\'evy-Khintchine form for each $x \in D$.
\end{rem}
\begin{proof}
See \cite{KST2011}.\end{proof}

In general, the parameters $(a,\alpha^i, b, \beta^i, c, \gamma^i, m,
\mu^i)_{i \in \set{1, \dotsc, d}}$ appearing in the representations of $F$ and $R$ in (\ref{Eq:Riccati_F}) and 
(\ref{Eq:Riccati_R}) have to satisfy
additional conditions, called the \emph{admissibility} conditions. These conditions
guarantee the existence of an affine Markov process $X$ with state space $D$ and with prescribed $F$ and $R$. It is clear that such conditions should depend strongly
on the geometry of the (boundary of the) state space $D$. Finding such (necessary and sufficient) conditions on the parameters for different types of
state spaces has been the focus of several publications. For $D =
\Rplus^m \times \RR^n$, the admissibility conditions were
derived in \cite{Duffie2003}. For the cone of
semi-definite matrices $D = S_d^+$, such conditions were found in \cite{Cuchiero2009}, and for symmetric irreducible cones, the admissibility conditions were found in \cite{CKMT2011}. Finally, for affine diffusions ($m = \mu^i = 0$) on polyhedral cones and on quadratic state spaces, the admissiblility conditions were given in \cite{Spreij2010}. 

\begin{defn}
We call the state space $D = \Rplus^m \times \RR^n$ with $ m,n \geq 0$ the \emph{canonical state space}.
\end{defn}

Affine processes on canonical state spaces are completely characterized in \cite{Duffie2003} in terms of the admissibility conditions imposed on $F$ and $R$. Affine processes on canonical state spaces have continuous trajectories (such processes are called continuous affine processes) if and only if the functions $F$ and $R$ satisfy the admissibility conditions and are polynomials of degree at most $2$ (see Proposition \ref{Prop:LevyK}).

\section{Homogenization procedure}\label{S:homben}

In this section, we consider continuous, affine processes on the canonical state space $D=\Rplus^m \times \RR^n$. We will next introduce a natural homogenization procedure, which allows to analyze the short-time asymptotics of the law of continuous affine processes. In the case of affine processes, the homogenization leads in fact to real analytic expansions with respect to the homogenization parameter. 

The following lemmas introduce the homogenization procedure.
\begin{lem}\label{L:hom}
Let $\psi: \cU \times \RR_{\geq0} \to \cU$ be the unique solution of the equation
$$\frac{\partial}{\partial t}\psi(u,t) = R \big(\psi(u,t)\big), \quad \psi(u,0)=u\in \cU,$$
where $R : \cU \rightarrow \CC^d$ is a quadratic polynomial. 
Then, for every $\epsilon > 0$, the function 
$$\psi^\epsilon(u,t):= \epsilon \psi \Big(\frac{u}{\epsilon}, \epsilon t\Big)$$
solves the equation
$$\frac{\partial}{\partial t}\psi^\epsilon(u,t) = R^\epsilon \big(\psi^\epsilon(u,t)\big), \quad \psi^\epsilon(u,0)=u$$
with $R^\epsilon (u):= \epsilon^2R\big(\epsilon^{-1}u\big)$ for $u \in \cU$.

Analogously, let $\phi: \cU \times \RR_{\geq0} \to \mathbb{C}$ be the unique solution of the equation
$$\frac{\partial}{\partial t}\psi(u,t) = F \big(\psi(u,t)\big), \quad \phi(u,0)=0. $$
Then, for every $\epsilon> 0$, the function 
$$\phi^\epsilon(u,t):= \epsilon \phi \Big(\frac{u}{\epsilon}, \epsilon t\Big)$$
solves the equation
$$\frac{\partial}{\partial t}\phi^\epsilon(u,t) = F^\epsilon \big(\psi^\epsilon(u,t)\big), \quad \phi^\epsilon(u,0)=0$$
with $F^\epsilon (u):= \epsilon^2F\big(\epsilon^{-1}u\big)$ for $u \in \cU$.
\end{lem}

The proof of Lemma \ref{L:hom} is simple, and we leave it as an exercise for the reader.
\begin{lem}\label{L:yield}
Under the previous assumptions, the limit
$\lim_{\epsilon \to 0}\psi^\epsilon=\psi^{(0)}$
exists uniformly on compact sets in $\cU \times \RR_{\geq0}$. Furthermore,
\begin{equation}
\psi^\epsilon (u,t) = \psi ^{(0)}(u,t)+\epsilon \psi^{(1)}(u,t)+\sum_{n\geq2} \epsilon^{n} \psi^{(n)}(u,t)
\label{E:coeffi}
\end{equation}
is a convergent power series expansion for small $\epsilon > 0$. The coefficient functions in (\ref{E:coeffi}) satisfy certain ordinary differential equations, i.e.,~in particular,
$$\frac{\partial}{\partial t}\psi^{(0)}(u,t) = R^{(0)} \big(\psi^{(0)}(u,t)\big), \quad \psi^{(0)}(u,0)=u \, ,$$
and
$$\frac{\partial}{\partial t}\psi^{(1)}(u,t) = \frac{\partial}{\partial \epsilon}\Big\vert_{\epsilon=0}R^{\epsilon} \big(\psi^{(0)}(u,t)\big) \psi^{(1)}(u,t), \quad \psi^{(1)}(u,0)=0.$$
For $n\ge 2$, the equations for the coefficient functions involve higher order derivatives. In complete analogy,
the limit $\lim_{\epsilon \to 0}\phi^\epsilon=\phi^{(0)}$
exists uniformly on compact sets in $\cU \times \RR_{\geq0}$. Furthermore
$$\phi^\epsilon (u,t) = \phi ^{(0)}(u,t)+\epsilon \phi^{(1)}(u,t)+\sum_{n\geq2} \epsilon^{u} \phi^{(n)}(u,t) \, ,$$
for small enough values of $ \epsilon $.
\end{lem}

\begin{proof}
Observe that $R^\epsilon = R^{(0)} + \epsilon R^{(1)} +\frac{\epsilon^2}{2}R^{(2)} $ and
$F^\epsilon = F^{(0)} + \epsilon F^{(1)} +\frac{\epsilon^2}{2}F^{(2)}.$
Hence, the vector fields appearing in the equation in Lemma \ref{L:yield} are polynomial in $ u $ and $ \epsilon$. Standard results on differential equations with polynomial vector fields yield the assertions in Lemma \ref{L:yield}, in particular, the real analyticity of the solution with respect to $\epsilon$.
\end{proof}

Let $X$ be an affine diffusion process with the corresponding functions $F$ and $R$. We can extend the solutions of the Riccati equations described above to maximal domains for $ u \in \mathbb{R}^d$, i.e., consider maximal local flows on $ \mathbb{R}^d $ with the vector fields $ F^\epsilon $ and $ R^\epsilon $. 
By $\hat{\Lambda}^{(i)}$, $i\ge 0$, are denoted the functions appearing in the following power series expansion in $\epsilon$:
\begin{equation}
\hat{\Lambda}^{(0)}(u) + \epsilon\hat{\Lambda}^{(1)}(u)+... := \phi^\epsilon(-u,1) + \langle x, \psi^\epsilon(-u,1) \rangle \, ,
\label{E:suppress}
\end{equation}
They are the solutions of the Riccati equations appearing in the previous lemmas. Note that we suppress the dependence on the initial value $x$ on the left-hand side of (\ref{E:suppress}). The functions $\hat{\Lambda}^{(i)}$ exist as extended real numbers for $ u \in \mathbb{R}^d$.
\begin{rem}\rm
If the expression on right-hand side of (\ref{E:suppress}) is finite, then the power series on the left-hand side converges absolutely for sufficiently small values of $ \epsilon $.
\end{rem}

\begin{rem}
For continuous affine processes, the homogenization procedure leads to the following representation:
\begin{align}
E \left[ \exp\left\{ - \langle \frac{u}{\epsilon}, X_\epsilon \rangle \right\} \right]& 
= \int_{D} \exp \left\{ - \langle \frac{u}{\epsilon}, z \rangle \rangle \right\} p_\epsilon (dz) \nonumber \\
& = \exp \left\{ \frac{\hat{\Lambda}^{(0)}(u)}{\epsilon} + \hat{\Lambda}^{(1)}(u) + ... \right\},
\label{E:stup}
\end{align}
where $ u $ is such that the expressions on both sides of (\ref{E:stup}) are finite for small enough values of 
$ \epsilon$.
\end{rem}

The representation in (\ref{E:stup}) valid for any continuous affine process was a motivation for us for introducing condition (\ref{E:missing}) used in the previous sections. However, the expansion in (\ref{E:stup}) is a little different from that in 
(\ref{E:missing}).

\section{Example: The Heston model}\label{S:Hesmod}
In this section, we find explicit formulas for the functions $\Lambda^{(i)}$, $0\le i\le 2$, associated with the  log-price process in the Heston model. Let us consider the following correlated Heston model: 
\begin{align}
dX_t = {} & (r + k V_t ) dt + \sqrt{V_t} dW_{1,\,t}, \nonumber \\
dV_t = {} & (a - bV_t) dt + \sigma \sqrt{V_t} dW_{2,\,t}, 
\label{E:inone}
\end{align}
where $r,\,k\in\mathbb{R}$, $a,\,b\ge 0$, $\sigma> 0$, and $W_{1,\,t}$ and $W_{2,\,t}$ are standard Brownian
motions with $d\langle W_1,W_2\rangle_t=\rho dt$. We assume that the correlation coefficient $\rho$ satisfies the condition $-1<\rho< 1$.
In (\ref{E:inone}), $X$ is the log-price process, and $V$ is the variance process. The initial conditions for 
the processes $X$ and $V$ are denoted by $x_0$ and $v_0$, respectively.
The Heston model was introduced in \cite{heston}. Note that in the present paper we consider the Heston model in which both the log-price and the variance equations contain drift terms generated by affine functions. Very often, e.g., in \cite{forjac:09,
fj,fjm,fjl,JR}, a special Heston model where $k=-\frac{1}{2}$ and $r=0$ is studied. 
An extended Heston model, in which the defining equations contain affine drift terms, is discussed in \cite{JM}.

The process $X$ is not an affine process. It is a projection of the two-dimensional affine process $(X,V)$ onto the first coordinate. The moment generating function of $X_t$ is given by
$
M_t(u) = \E{\exp\{ u X_t \}} = \exp\left\{ C(u,t) + D(u,t) v_0 + u x_0 \right\},
$
where
$$
C(u,t)=rut+\frac{a}{\sigma^2}\bigg[(b-\rho\sigma u+d(u))t-2\log\Big(\frac{1-g(u)e^{d(u)t}}{1-g(u)}\Big)\bigg], 
$$
$$
D(u,t)=\frac{b+d(u)-\rho\sigma u}{\sigma^2} \bigg(\frac{1-e^{d(u)t}}{1-g(u)e^{d(u)t}}\bigg),
$$
$$
g(u)=\frac{b-\rho\sigma u+d(u)}{b-\rho\sigma u-d(u)},
$$
and
$$
d(u)=\sqrt{(\rho\sigma u - b)^2 - \sigma^2(2ku+u^2)}
$$
(see \cite{as}). Here and in the sequel, the symbol $\sqrt{\cdot}$ stands for the principal square root function. We will explain below the meaning of the logarithmic function appearing in the expression for the 
function $C$ (see the discussion after formula (\ref{E:nos1})). Note that for $u=0$, the expressions for the functions $C$ and $D$ should be understood in the limiting sense. More precisely, 
$$
C(0,t)=\lim_{u\rightarrow 0}C(u,t)=0\quad\mbox{and}\quad D(0,t)=\lim_{u\rightarrow 0}D(u,t)=0
$$ 
for all $t> 0$.

It is clear that
$$
\E{\exp\{-\frac{u}{t}X_t \}} = \exp\left\{ C(-\frac{u}{t},t) + D(-\frac{u}{t},t) v_0 -\frac{u}{t}x_0 \right\}.
$$
Denote 
$\Lambda(u,t)=t\log\E{\exp\{-\frac{u}{t}X_t \}}$.
Then
\begin{equation}
\Lambda(u,t)=tC(-\frac{u}{t},t)+tD(-\frac{u}{t},t) v_0-ux_0.
\label{E:endeq}
\end{equation}
Next, set $A(u)=b-\rho\sigma u$. It is not hard to see that
\begin{align*}
&D(u,t)=\frac{1}{\sigma^2}(A(u)+d(u))\frac{1-e^{d(u)t}}{1-\frac{A(u)+d(u)}
{A(u)-d(u)}e^{d(u)t}} \\
&=\frac{1}{\sigma^2}(A(u)^2-d(u)^2)\frac{\sinh\frac{d(u)t}{2}}
{d(u)\cosh\frac{d(u)t}{2}+A(u)\sinh\frac{d(u)t}{2}}. 
\end{align*}
Moreover, 
\begin{align*}
C(u,t)&=rut+\frac{a}{\sigma^2}\Bigg[(A(u)+d(u))t-2\log \Bigg(\frac{1-\frac{A(u)+d(u)}{A(u)-d(u)}
e^{d(u)t}}{1-\frac{A(u)+d(u)}{A(u)-d(u)}}\Bigg)\Bigg]\\
&= rut+\frac{a}{\sigma^2}\Bigg[A(u)t-2\log \frac{d(u)\cosh\frac{d(u)t}{2}+A(u)\sinh\frac{d(u)t}{2}}{d(u)}\Bigg].
\end{align*}

Using the previous formula, we obtain
\begin{align}
&C\Big(-\frac{u}{t},t\Big)=-ru \nonumber \\
&+\frac{a}{\sigma^2}\Bigg[bt+\rho\sigma u-2 \log \frac{d(-\frac{u}{t})t
\cosh\frac{d(-\frac{u}{t})t}{2}+(bt+\rho\sigma u)\sinh\frac{d(-\frac{u}{t})t}{2}}{d(-\frac{u}{t})t}\Bigg].
\label{E:kem33}
\end{align}
We also have
$$
A\Big(-\frac{u}{t}\Big)=b+\rho\sigma \frac{u}{t}, 
$$
$$
A^2\Big(-\frac{u}{t}\Big)=b^2+2b\rho\sigma \frac{u}{t} + \rho^2 \sigma^2 \frac{u^2}{t^2},
$$
\begin{align*}
d^2\Big(-\frac{u}{t}\Big)={} & -\frac{u^2(1-\rho^2)\sigma^2}{t^2}+\frac{2\sigma u(k\sigma+b\rho)}{t}+b^2,\\
\end{align*}
\begin{align*}
&\frac{1}{\sigma^2}\bigg(A^2\Big(-\frac{u}{t}\Big)-d^2\Big(-\frac{u}{t}\Big)\bigg)
= \frac{u^2}{t^2}-\frac{2ku}{t},
\end{align*}
and
\begin{align}
&D\Big(-\frac{u}{t},t\Big)=\Big(\frac{u^2}{t^2} -\frac{2ku}{t}\Big)\frac{\sinh\frac{d\left(-\frac{u}{t}\right)t}{2}}{d\left(-\frac{u}{t}\right)\cosh\frac{d\left(-\frac{u}{t}\right)t}{2}
+A\left(-\frac{u}{t}\right)\sinh\frac{d\left(-\frac{u}{t}\right)t}{2}}. 
\label{E:no1}
\end{align}

Let us denote by $Z$ the set of such real numbers $u$ that the expressions on the right-hand side of 
(\ref{E:kem33}) and (\ref{E:no1}) are finite for all small enough values of $t$, and put
$$
\hat{S}(u,t)=d\Big(-\frac{u}{t}\Big)\frac{t}{2}.
$$
It is easy to see that
$$
\hat{S}(u,t)=\frac{1}{2}\sqrt{-u^2(1-\rho^2)\sigma^2+2tu(k\sigma^2+b\rho\sigma)+t^2b^2}.
$$
In the previous formula, $t$ is a real number. Therefore, for every real number $u\neq 0$, $\hat{S}(u,t)$ is purely imaginary for 
all numbers $t$ with $|t|$ small enough. For such $u$ and $t$,
$\hat{S}(u,t)=iS(u,t)$, where
\begin{equation}
S(u,t)=\frac{1}{2}\sqrt{u^2(1-\rho^2)\sigma^2-2tu(k\sigma^2+b\rho\sigma)-t^2b^2}
\label{E:ut}
\end{equation}
is a real number. It follows that
\begin{align}
&tC\Big(-\frac{u}{t},t\Big)
=-tru \nonumber \\
&+t\frac{a}{\sigma^2}\Bigg[bt+\rho\sigma u-2 \log \frac{2S(u,t)
\cos S(u,t)+(bt+\rho\sigma u)\sin S(u,t)}{2S(u,t)}\Bigg].
\label{E:kems3}
\end{align}
and 
\begin{align}
&tD\Big(-\frac{u}{t},t\Big)
=\Big(u^2 -2tku\Big)\frac{\sin S(u,t)}
{2S(u,t)\cos S(u,t)
+(bt+\rho\sigma u)\sin S(u,t)}.
\label{E:nos1}
\end{align}

Our next goal is to introduce an additional condition under which the logarithmic function 
appearing in formula (\ref{E:kems3}) exists, and the expressions on the right-hand sides of 
(\ref{E:kems3}) and (\ref{E:nos1}) are finite. Recall that we have assumed that $u\neq 0$ and $|t|$ is small enough.
Set 
$$
\widetilde{S}(u)=\lim_{t\rightarrow 0}\left[2S(u,t)\cos S(u,t)+(bt+\rho\sigma u)\sin S(u,t)\right].
$$
Then, we have
$$
\lim_{t\rightarrow 0}S(u,t)=\frac{1}{2}|u|\sigma\sqrt{1-\rho^2}
$$
and
$$
\widetilde{S}(u)=|u|\sigma\sqrt{1-\rho^2}\cos\frac{|u|\sigma\sqrt{1-\rho^2}}{2}+u\sigma\rho 
\sin\frac{|u|\sigma\sqrt{1-\rho^2}}{2},
$$

Let $\rho\neq 0$, and assume that
\begin{equation}
-\frac{2}{\sigma\sqrt{1-\rho^2}}\arctan\frac{\sqrt{1-\rho^2}}{\rho}< u<
\frac{2}{\sigma\sqrt{1-\rho^2}}\left(\pi-\arctan\frac{\sqrt{1-\rho^2}}{\rho}\right).
\label{E:condd1}
\end{equation}
The restriction in (\ref{E:condd1}) means that the variable $u$ 
is bounded from below by the largest negative root
of the function 
$$
\widehat{S}(u)=\sqrt{1-\rho^2}\cos\frac{u\sigma\sqrt{1-\rho^2}}{2}+\rho 
\sin\frac{u\sigma\sqrt{1-\rho^2}}{2},
$$
and from above by the smallest positive root of the same function. Note that $\widehat{S}(0)> 0$.
Therefore, we have $\widehat{S}(u)> 0$, for all $u$ satisfying the condition in (\ref{E:condd1}).

It is easy to see that $\widetilde{S}(u)=\sigma|u|\widehat{S}(u)$ for all $u\neq 0$, satisfying the condition
in (\ref{E:condd1}). Hence, $\widetilde{S}(u)> 0$, under the same restrictions on $u$. It follows from 
(\ref{E:nos1}) that for all $u\neq 0$ such that (\ref{E:condd1}) holds, the right-hand side of 
(\ref{E:nos1}) is eventually finite as $t\rightarrow 0$, and moreover
\begin{equation}
\lim_{t\rightarrow 0}tD\Big(-\frac{u}{t},t\Big)=\frac{u\sin\frac{u\sigma\sqrt{1-\rho^2}}{2}}
{\sigma\left(\sqrt{1-\rho^2}\cos\frac{u\sigma\sqrt{1-\rho^2}}{2}+\rho\sin\frac{u\sigma\sqrt{1-\rho^2}}{2}
\right)}.
\label{E:li1}
\end{equation}
In addition, the expression under the logarithm sign in (\ref{E:kems3}) is eventually positive, and
\begin{equation}
\lim_{t\rightarrow 0}tC\Big(-\frac{u}{t},t\Big)=0,
\label{E:li2}
\end{equation}

In the case where $\rho=0$, the condition in (\ref{E:condd1}) becomes
\begin{equation}
-\frac{\pi}{\sigma}< u<\frac{\pi}{\sigma}.
\label{E:condd2}
\end{equation}
The analysis here proceeds similarly to that in the previous case.

The next statement provides explicit expressions for the function $\Lambda^{(0)}$. This statement was obtained
in \cite{fj} (see formula (2) in \cite{fj}, see also \cite{fjl}) in a special case where $k=-\frac{1}{2}$ and $r=0$. 
\begin{thm}\label{T:te1}
Suppose $\rho\neq 0$ and condition (\ref{E:condd1}) holds.
Then $u\in Z$ and the following formula is valid:
\begin{equation}
\Lambda^{(0)}(u)=\frac{v_0u\sin\frac{u\sigma\sqrt{1-\rho^2}}{2}}{\sigma\Big(\sqrt{1-\rho^2}
\cos\frac{u\sigma\sqrt{1-\rho^2}}{2}+\rho \sin\frac{u\sigma\sqrt{1-\rho^2}}{2}\Big)}-x_0u.
\label{E:fo1}
\end{equation}
If $\rho=0$ and condition (\ref{E:condd2}) holds, then $u\in Z$ and
$$
\Lambda^{(0)}(u)=\frac{v_0u}{\sigma}\tan\frac{u\sigma}{2}-x_0u.
$$\end{thm}

Theorem \ref{T:te1} follows from (\ref{E:endeq}), (\ref{E:li1}), and (\ref{E:li2}).

Recall that for $x\in\mathbb{R}$, the critical point $u^{*}(x)$ is the solution of the equation $\partial_u\Lambda^{(0)}(u)=-x$. Put $\theta=\frac{\sigma\sqrt{1-\rho^2}}{2}$. Then, using (\ref{E:fo1}), we obtain
$$
\partial_u\Lambda^{(0)}(u)=\frac{v_0}{2\sigma}\frac{\rho[1-\cos(2\theta u)]+\sqrt{1-\rho^2}
\sin(2\theta u)+\sigma(1-\rho^2)u}{(\sqrt{1-\rho^2}\cos(\theta u)+\rho\sin(\theta u))^2}-x_0.
$$
In a special case where $\rho=0$, we have
$$
\partial_u\Lambda^{(0)}(u)=\frac{v_0}{\sigma}\frac{\sin(2\theta u)+\sigma u}{1+\cos(2\theta u)}-x_0.
$$

In the following two statements, we provide formulas for the critical point $u^{*}(x)$ and the second derivative of the
function $\Lambda^{(0)}(u)$. These results can be used in the asymptotic formulas established in the previous sections in the case of the Heston model.
\begin{lem}\label{L:cri}
Suppose $\rho\neq 0$ and condition (\ref{E:condd1}) holds. Then, for every $x\in\mathbb{R}$, the critical point 
$u^{*}(x)$ is the unique solution to the equation
$$
\frac{\rho[1-\cos(2\theta u)]+\sqrt{1-\rho^2}
\sin(2\theta u)+\sigma(1-\rho^2)u}{(\sqrt{1-\rho^2}\cos(\theta u)+\rho\sin(\theta u))^2}=\frac{2\sigma}{v_0}(x_0-x).
$$
If $\rho=0$ and condition (\ref{E:condd2}) holds, then for every $x\in\mathbb{R}$,  
$u^{*}(x)$ is the unique solution to the equation
$$
\frac{\sin(2\theta u)+\sigma u}{1+\cos(2\theta u)}=\frac{\sigma}{v_0}(x_0-x).
$$
\end{lem}
\begin{lem}\label{L:cria}
Suppose $\rho\neq 0$ and condition (\ref{E:condd1}) holds. Then
$$
\partial^2\Lambda^{(0)}(u) 
=\frac{v_0S(u)}{2\sigma[\sqrt{1-\rho^2}\cos(\theta u)+\rho\sin(\theta u)]^3}
$$
where
\begin{align*}
S(u)&=(2\theta+\sigma\sqrt{1-\rho^2})
[\rho\sqrt{1-\rho^2}\sin(\theta u)+(1-\rho^2)\cos(\theta u)] \\
&\quad+2\sigma\theta(1-\rho^2)u[\sqrt{1-\rho^2}\sin(\theta u)-\rho\cos(\theta u)].
\end{align*}
If $\rho=0$ and condition (\ref{E:condd2}) holds, then 
$$
\partial^2\Lambda^{(0)}(u)=\frac{v_0}{2\sigma}
\frac{(2\theta+\sigma)\cos(\theta u)+2\theta\sigma u\sin(\theta u)}{\cos^3(\theta u)}.
$$
\end{lem}

Lemmas \ref{L:cri} and \ref{L:cria} are straightforward, and their proofs are omitted.

We will next compute the functions $\Lambda^{(1)}$ and $\Lambda^{(2)}$.
Recall that 
$$
\int \exp \big( -\frac{u}{t} z \big) p_t(dz) = \exp\big( \frac{\Lambda^{(0)}(u)}{t} \big) \exp \big( \Lambda^{(1)}(u) \big) \bigg( 1 + t\Lambda^{(2)}(u) + \ldots \bigg).
$$
Therefore,
$$
\Lambda(u,t)=\Lambda^{(0)}(u)+t\Lambda^{(1)}(u)+t\log(1 + t\Lambda^{(2)}(u) + \ldots ).
$$
By differentiating the previous formula with respect to $t$, we obtain
$$
\Lambda^{(1)}(u)=\lim_{t\rightarrow 0}\frac{\partial\Lambda}{\partial t}(u,t)
$$
and
\begin{equation}
\Lambda^{(2)}(u)=\frac{1}{2}\lim_{t\rightarrow 0}\frac{\partial^2\Lambda}{\partial t^2}(u,t).
\label{E:par2}
\end{equation}

Let us fix $u\neq 0$ such as in Theorem \ref{T:te1}. Then the function $t\mapsto S(u,t)$ defined by (\ref{E:ut})
is real analytic in $t$ in a small neighborhood of $t=0$, depending on $u$. Using the Taylor formula, we obtain
\begin{equation}
S(u,t)=c_0(u)+c_1(u)t+\frac{1}{2}c_2(t)t^2+O\left(t^3\right)
\label{E:son1}
\end{equation}
as $t\rightarrow 0$, where the $O$-estimate depends on $u$, and the coefficients are given by
\begin{equation}
c_0(u)=\frac{|u|\sigma}{2}\sqrt{1-\rho^2},
\label{E:son2}
\end{equation}
\begin{equation}
c_1(u)=-\frac{|u|}{u}\frac{k\sigma+b\rho}{2\sqrt{1-\rho^2}},
\label{E:son3}
\end{equation}
and
$$
c_2(u)=-\frac{|u|}{u^2}\frac{b^2(1-\rho^2)
+(k\sigma+b\rho)^2}{2\sigma(1-\rho^2)^{\frac{3}{2}}}.
$$

Our next goal is to expand the functions $t\mapsto\sin S(u,t)$ 
and $t\mapsto\cos S(u,t)$. Using the Taylor formula and (\ref{E:son1}), we get
\begin{equation}
\sin S(u,t)=U_0(u)+U_1(u)t+\frac{1}{2}U_2(u)t^2+O(t^3)
\label{E:sono1}
\end{equation}
as $t\rightarrow 0$, where 
\begin{equation}
U_0(u)=\sin c_0(u),
\label{E:U0}
\end{equation}
\begin{equation}
U_1(u)=c_1(u)\cos c_0(u),
\label{E:U1}
\end{equation}
and
$$
U_2(u)=c_2(u)\cos c_0(u)-c_1(u)^2\sin c_0(u).
$$
Similarly,
\begin{equation}
\cos S(u,t)=W_0(u)+W_1(u)t+\frac{1}{2}W_2(u)t^2+O(t^3)
\label{E:sono2}
\end{equation}
as $t\rightarrow 0$, where 
$$
W_0(u)=\cos c_0(u),
$$
$$
W_1(u)=-c_1(u)\sin c_0(u),
$$
and
$$
W_2(u)=-[c_2(u)\sin c_0(u)+c_1(u)^2\cos c_0(u)].
$$

We will next expand the functions $t\mapsto tD(-\frac{u}{t},t)$ and $t\mapsto tC(-\frac{u}{t},t)$. It follows from
(\ref{E:son1}), (\ref{E:sono1}), and (\ref{E:sono2}) that
\begin{equation}
2S(u,t)\cos S(u,t)+(bt+\rho\sigma u)\sin S(u,t)=V_0(u)+V_1(u)t+\frac{1}{2}V_2(u)t^2+O(t^3)
\label{E:sono3}
\end{equation}
as $t\rightarrow 0$, where
$$
V_0(u)=2c_0(u)W_0(u)+\rho\sigma uU_0(u),
$$
$$
V_1(u)=2c_0(u)W_1(u)+2c_1(u)W_0(u)+bU_0(u)+\rho\sigma uU_1(u),
$$
and
$$
V_2(u)=2c_0(u)W_2(u)+4c_1(u)W_1(u)+2c_2(u)W_0(u)+2bU_1(u)+\rho\sigma uU_2(u).
$$
It is not hard to see that
\begin{equation}
V_0(u)=2c_0(u)\cos c_0(u)+\rho\sigma u\sin c_0(u),
\label{E:V0}
\end{equation}
\begin{align}
V_1(u)&=(2+\rho\sigma u)c_1(u)\cos c_0(u) \nonumber \\
&\quad+(b-2c_0(u)c_1(u))\sin c_0(u),
\label{E:V1}
\end{align}
and
\begin{align*}
V_2(u)&=[2c_2(u)+2bc_1(u)+\rho\sigma uc_2(u)-2c_0(u)c_1(u)^2]\cos c_0(u) \\
&\quad-[2c_0(u)c_2(u)+4c_1(u)^2+\rho\sigma uc_1(u)^2]\sin c_0(u).
\end{align*}

Therefore,
\begin{equation}
tD(-\frac{u}{t},t)=(u^2-2tku)\frac{U_0(u)+U_1(u)t+\frac{1}{2}U_2(u)t^2+O(t^3)}
{V_0(u)+V_1(u)t+\frac{1}{2}V_2(u)t^2+O(t^3)}
\label{E:sono4}
\end{equation}
as $t\rightarrow 0$ (see (\ref{E:nos1}), (\ref{E:sono1}), and (\ref{E:sono3})).

Set
\begin{equation}
tD(-\frac{u}{t},t)=T_0(u)+T_1(u)t+\frac{1}{2}T_2(u)t^2+O(t^3)
\label{E:sono5}
\end{equation}
as $t\rightarrow 0$. Then, (\ref{E:sono4}) and (\ref{E:sono5}) give
$$
T_0(u)V_0(u)=u^2U_0(u),
$$
$$
T_0(u)V_1(u)+T_1(u)V_0(u)=u^2U_1(u)-2kuU_0(u),
$$
and
$$
\frac{1}{2}T_0(u)V_2(u)+T_1(u)V_1(u)+\frac{1}{2}T_2(u)V_0(u)
=\frac{1}{2}u^2U_2(u)-2kuU_1(u).
$$
It follows from the previous equalities that
$$
T_0(u)=\frac{u^2U_0(u)}{V_0(u)},
$$
$$
T_1(u)=\frac{u^2U_1(u)V_0(u)-2kuU_0(u)V_0(u)-u^2U_0(u)V_1(u)}
{V_0(u)^2},
$$
and
$$
T_2(u)=\frac{Q(u)}{V_0(u)^3},
$$
where
\begin{align}
Q(u)&=u^2U_2(u)V_0(u)^2-4kuU_1(u)V_0(u)^2-u^2U_0(u)V_0(u)V_2(u) \nonumber \\
&\quad-2u^2U_1(u)V_0(u)V_1(u)+4kuU_0(u)V_0(u)V_1(u)+2u^2U_0(u)V_1(u)^2.
\label{E:Q}
\end{align}
Therefore, 
the following asymptotic formula:
\begin{align}
tD(-\frac{u}{t},t)&=\frac{u^2U_0(u)}{V_0(u)}+t\frac{u^2U_1(u)V_0(u)-2kuU_0(u)V_0(u)-u^2U_0(u)V_1(u)}
{V_0(u)^2} \nonumber \\
&+\frac{t^2}{2}\frac{Q(u)}{V_0(u)^3}+O(t^3)
\label{E:asfa}
\end{align}
as $t\rightarrow 0$.

Now, we turn our attention to the function $t\mapsto tC(-\frac{u}{t},t)$. Using (\ref{E:kems3}), we
see that
\begin{align}
tC(-\frac{u}{t},t)=-tru+t\frac{a}{\sigma^2}\Bigg[bt+\rho\sigma u
-2 \log\frac{V_0(u)+V_1(u)t+O(t^2)}{2c_0(u)+2c_1(u)t+O\left(t^2\right)}\Bigg].
\label{E:co1}
\end{align}
Set
$$
\frac{V_0(u)+V_1(u)t+O(t^2)}{2c_0(u)+2c_1(u)t+O\left(t^2\right)}
=L_0(u)+L_1(u)t+O(t^2)
$$
as $t\rightarrow 0$. It is not hard to see that
\begin{equation}
L_0(u)=\frac{V_0(u)}{2c_0(u)},
\label{E:lis1}
\end{equation}
\begin{equation}
L_1(u)=\frac{c_0(u)V_1(u)-c_1(u)V_0(u)}{2c_0(u)^2}.
\label{E:lis2}
\end{equation}
We also have
\begin{equation}
\log [L_0(u)+L_1(u)t+O(t^2)]=\log L_0(u)+\frac{L_1(u)}{L_0(u)}t+O(t^2)
\label{E:li3}
\end{equation}
as $t\rightarrow 0$. It follows from (\ref{E:co1})-(\ref{E:li3}) that
\begin{align}
tC(-\frac{u}{t},t)&=\left[\frac{a\rho u}{\sigma}-ru-\frac{2a}{\sigma^2}\log\frac{V_0(u)}{2c_0(u)}\right]t
\nonumber \\
&\quad+\left[\frac{ab}{\sigma^2}-\frac{2a}{\sigma^2}\frac{c_0(u)V_1(u)-c_1(u)V_0(u)}{c_0(u)V_0(u)}\right]t^2+O(t^3)
\label{E:li4}
\end{align}
as $t\rightarrow 0$.

Next, we will find explicit expressions for the functions $\Lambda^{(1)}$
and $\Lambda^{(2)}$. Suppose $\rho\neq 0$, $u\neq 0$, and condition (\ref{E:condd1}) holds.
Then
\begin{align}
&\Lambda^{(1)}(u)=\left(\frac{a\rho}{\sigma}-r\right)u-\frac{2a}{\sigma^2}\log\frac{V_0(u)}{2c_0(u)}
\nonumber \\
&\quad+v_0\frac{u^2U_1(u)V_0(u)-2kuU_0(u)V_0(u)-u^2U_0(u)V_1(u)}{V_0(u)^2}.
\label{E:lala}
\end{align}
Formula (\ref{E:lala}) can be established, using (\ref{E:asfa}) and (\ref{E:li4}).

The next statement provides an explicit expression for the function $\Lambda^{(1)}$ in terms of the Heston model parameters.
\begin{thm}\label{T:coco1}
Suppose $\rho\neq 0$, $u\neq 0$, and condition (\ref{E:condd1}) holds. Then
\begin{align}
&\Lambda^{(1)}(u)=\left(\frac{a\rho}{\sigma}-r\right)u-\frac{2a}{\sigma^2}
\log\frac{\sqrt{1-\rho^2}\cos\frac{u\sigma\sqrt{1-\rho^2}}{2}
+\rho\sin\frac{u\sigma\sqrt{1-\rho^2}}{2}}{\sqrt{1-\rho^2}} 
\nonumber \\
&+v_0\frac{E_1(u)\cos^2\frac{u\sigma\sqrt{1-\rho^2}}{2}+E_2(u)\cos\frac{u\sigma\sqrt{1-\rho^2}}{2}\sin\frac{u\sigma\sqrt{1-\rho^2}}{2}+E_3(u)\sin^2\frac{u\sigma\sqrt{1-\rho^2}}{2}}
{\sigma^2\left(\sqrt{1-\rho^2}\cos\frac{u\sigma\sqrt{1-\rho^2}}{2}
+\rho\sin\frac{u\sigma\sqrt{1-\rho^2}}{2}\right)^2},
\label{E:ll}
\end{align}
where
$$
E_1(u)=-\frac{u\sigma(k\sigma+b\rho)}{2},
$$
$$
E_2(u)=-2k\sigma\sqrt{1-\rho^2}
+\frac{k\sigma+b\rho}{\sqrt{1-\rho^2}},
$$
and
$$
E_3(u)=-\left(2k\rho\sigma+b+u\sigma\frac{k\sigma+b\rho}{2}\right).
$$
If $\rho=0$, then formula (\ref{E:ll}) holds for all $u$ satisfying condition (\ref{E:condd2}).
\end{thm}

\it Proof. \rm Taking into account (\ref{E:lala}), (\ref{E:son2}), (\ref{E:son3}), (\ref{E:U0}),
(\ref{E:U1}), (\ref{E:V0}), (\ref{E:V1}), we obtain
\begin{align*}
&\Lambda^{(1)}(u)=\left(\frac{a\rho}{\sigma}-r\right)u-\frac{2a}{\sigma^2}
\log\frac{|u|\sigma\sqrt{1-\rho^2}\cos\frac{|u|\sigma\sqrt{1-\rho^2}}{2}
+\rho\sigma u\sin\frac{|u|\sigma\sqrt{1-\rho^2}}{2}}{|u|\sigma\sqrt{1-\rho^2}} 
\\
&+v_0\frac{\widetilde{E}_1(u)\cos^2\frac{|u|\sigma\sqrt{1-\rho^2}}{2}+\widetilde{E}_2(u)
\cos\frac{|u|\sigma\sqrt{1-\rho^2}}{2}\sin\frac{|u|\sigma\sqrt{1-\rho^2}}{2}+\widetilde{E}_3(u)\sin^2\frac{|u|\sigma\sqrt{1-\rho^2}}{2}}
{\left(|u|\sigma\sqrt{1-\rho^2}\cos\frac{|u|\sigma\sqrt{1-\rho^2}}{2}
+\rho\sigma u\sin\frac{|u|\sigma\sqrt{1-\rho^2}}{2}\right)^2},
\end{align*}
where
$$
\widetilde{E}_1(u)=-\frac{u^3}{2}\sigma(k\sigma+b\rho),
$$
$$
\widetilde{E}_2(u)=u\left[-\rho\sigma u|u|\frac{k\sigma+b\rho}{2\sqrt{1-\rho^2}}
-2k|u|\sigma\sqrt{1-\rho^2}+(2+\rho\sigma u)|u|\frac{k\sigma+b\rho}{2\sqrt{1-\rho^2}}\right],
$$
and
$$
\widetilde{E}_3(u)=-u^2\left[2k\rho\sigma+b+u\sigma\frac{k\sigma+b\rho}{2}\right].
$$
Next, replacing $|u|$ by $u$ in the previous formulas (it is not hard to see that this can be done)
and making several cancellations, we obtain formula (\ref{E:ll}).

This completes the proof of Theorem \ref{T:coco1}.

Our final goal in the present section is to find an explicit formula for the function $\Lambda^{(2)}$
in terms of the Heston model parameters. It follows from (\ref{E:endeq}), (\ref{E:par2}),
(\ref{E:asfa}), and (\ref{E:li4}) that
\begin{equation}
\Lambda^{(2)}(u)=\frac{ab}{\sigma^2}-\frac{2a}{\sigma^2}\frac{c_0(u)V_1(u)-c_1(u)V_0(u)}{c_0(u)V_0(u)}
+\frac{v_0}{2}\frac{Q(u)}{V_0(u)^3},
\label{E:finni}
\end{equation}
where $Q(u)$ is given by (\ref{E:Q}). Now, it is clear how to obtain an explicit expression for the 
function $\Lambda^{(2)}$, expressed in terms of the Heston model parameters.
It suffices to transform the formula in (\ref{E:finni}), using the explicit expressions 
for the functions $c_i$, $U_i$, $V_i$ with $i=0,1,2$, and the function $Q$. Let us also note that the value of 
the function on the right-hand of
formula (\ref{E:finni}) does not change if we replace $|u|$ by $u$. Taking into account what was said above, and making long but straightforward computations, we see that the following statement holds.
\begin{thm}\label{T:fin}
Suppose $\rho\neq 0$, $u\neq 0$, and condition (\ref{E:condd1}) holds. Then
\begin{align}
&\Lambda^{(2)}(u)=\frac{ab}{\sigma^2}-\frac{a}{\sigma^3(1-\rho^2)u}\frac{I_0(u)}
{\sqrt{1-\rho^2}\cos\frac{u\sigma\sqrt{1-\rho^2}}{2}+\rho
\sin\frac{u\sigma\sqrt{1-\rho^2}}{2}} \nonumber \\
&+\frac{v_0}{2}\frac{I_1(u)}{\sigma^2\left[\sqrt{1-\rho^2}\cos\frac{u\sigma\sqrt{1-\rho^2}}
{2}+\rho\sin\frac{u\sigma\sqrt{1-\rho^2}}{2}\right]^2} \nonumber \\
&+\frac{v_0}{2}\frac{I_2(u)I_3(u)}{u\sigma^3\left[\sqrt{1-\rho^2}\cos\frac{u\sigma\sqrt{1-\rho^2}}
{2}+\rho\sin\frac{u\sigma\sqrt{1-\rho^2}}{2}\right]^3},
\label{E:finik1}
\end{align}
where
\begin{align*}
I_0(u)&=\left[-u\rho\sigma\sqrt{1-\rho^2}(k\sigma+b\rho)\right]\cos\frac{u\sigma\sqrt{1-\rho^2}}{2} \\
&\quad+\left[2b+2\rho k\sigma+u\sigma(k\sigma+b\rho)(1-\rho^2)\right]\sin\frac{u\sigma\sqrt{1-\rho^2}}{2},
\end{align*}
\begin{align*}
I_1(u)&=-\frac{u\left[b^2(1-\rho^2)+(k\sigma+b\rho)^2\right]}{2(1-\rho^2)}
+\frac{(k\sigma+b\rho)^2}{1-\rho^2}\sin^2\frac{u\sigma\sqrt{1-\rho^2}}{2} \\
&\quad+\left[\frac{b^2(1-\rho^2)+(k\sigma+b\rho)^2}{\sigma(1-\rho^2)^{\frac{3}{2}}}
+\frac{b(k\sigma+b\rho)}{\sqrt{1-\rho^2}}\right]\sin\frac{u\sigma\sqrt{1-\rho^2}}{2}
\cos\frac{u\sigma\sqrt{1-\rho^2}}{2},
\end{align*}
\begin{align*}
I_2(u)&=2\left[2k\sigma\sqrt{1-\rho^2}-\frac{(2+\rho\sigma u)(k\sigma+b\rho)}{2\sqrt{1-\rho^2}}
\right]\cos\frac{u\sigma\sqrt{1-\rho^2}}{2} \\
&\quad+2\left[2k\rho\sigma+b+\frac{u\sigma(k\sigma+b\rho)}{2}\right]\sin\frac{u\sigma\sqrt{1-\rho^2}}{2},
\end{align*}
and
\begin{align*}
I_3(u)&=-u\sigma\sqrt{1-\rho^2}+b\sin^2\frac{u\sigma\sqrt{1-\rho^2}}{2} \\
&\quad-\frac{k\sigma+b\rho}{\sqrt{1-\rho^2}}\sin\frac{u\sigma\sqrt{1-\rho^2}}{2}
\cos\frac{u\sigma\sqrt{1-\rho^2}}{2}.
\end{align*}

If $\rho=0$, then formula (\ref{E:finik1}) holds for all $u$ satisfying condition (\ref{E:condd2}).
\end{thm}

\it Proof. \rm The second term on the right-hand side of (\ref{E:finik1}) can be obtained from the corresponding term in (\ref{E:finni}) by taking into account (\ref{E:son2}), (\ref{E:son3}), (\ref{E:V0}), and (\ref{E:V1}). Next, using (\ref{E:Q}), we see that
\begin{align}
\frac{Q(u)}{V_0(u)^3}&=\frac{u^2\left[U_2(u)V_0(u)-U_0(u)V_2(u)\right]}{V_0(u)^2} \nonumber \\
&\quad+\frac{\left[4kuV_0(u)+2u^2V_1(u)\right]\left[U_0(u)V_1(u)-U_1(u)V_0(u)\right]}
{V_0(u)^3}.
\label{E:frr}
\end{align}
Moreover
\begin{align*}
U_2(u)V_0(u)-U_0(u)V_2(u)&=2c_0(u)c_2(u)+4c_1(u)^2\sin^2c_0(u) \\
&\quad-2\left[c_2(u)+bc_1(u)\right]\sin c_0(u)\cos c_0(u),
\end{align*}
\begin{align*}
4kuV_0(u)+2u^2V_1(u)&=2u\left[4kc_0(u)+u(2+\rho\sigma u)c_1(u)\right]\cos c_0(u) \\
&\quad+2u^2\left[2k\rho\sigma+b-2c_0(u)c_1(u)\right]\sin c_0(u),
\end{align*}
and
\begin{align*}
U_0(u)V_1(u)-U_1(u)V_0(u)&=b\sin^2c_0(u)-2c_0(u)c_1(u) \\
&\quad+2c_1(u)\sin c_0(u)\cos c_0(u).
\end{align*}

Set 
$$
I_1(u)=U_2(u)V_0(u)-U_0(u)V_2(u),
$$
$$
I_2(u)=u^{-2}\left[4kuV_0(u)+2u^2V_1(u)\right],
$$
and
$$
I_3(u)=U_0(u)V_1(u)-U_1(u)V_0(u).
$$
Next, taking into account (\ref{E:finni}), (\ref{E:frr}), and using the explicit expressions for the functions $c_i$, $U_i$, and $V_1$ with $i=0,1,2$, which were found above, we obtain (\ref{E:finik1}).

This completes the proof of Theorem \ref{T:fin}.
\begin{rem}\label{R:ho} \rm The present remark concerns the continuity of the functions $\Lambda^{(i)}$
with $i=0,1,2$ on their domain. Recall that 
$\Lambda^{(i)}(0)=0$. It follows from Theorems \ref{T:te1}, \ref{T:coco1}, and \ref{T:fin} that the functions
$\Lambda^{(i)}$ are continuous on their domain with a possible exception of the point $u=0$. However, it 
is not hard to see, using
the explicit expressions for the functions $\Lambda^{(i)}$, provided in the theorems mentioned above, that
$$
\lim_{u\rightarrow 0}\Lambda^{(i)}(u)=0\quad\mbox{for}\quad i=0,1,2.
$$
\end{rem}

\end{document}